\documentclass[letterpaper,11pt]{article}
 
\usepackage{microtype}
\usepackage{amsmath}
\usepackage{fullpage}
\usepackage{amsthm}
\usepackage{lmodern}
\usepackage[T1]{fontenc}
\usepackage{hyperref}

\newtheorem{theorem}{Theorem}
\newtheorem{lemma}{Lemma}

\def\hl{{\sc HL}}
\def\hhl{{\sc HHL}}
\def\dist{{\rm dist}}

\def\dist{{\mbox{\rm dist}}}

\newcommand{\OPT}{\mathrm{OPT}}
\newcommand{\LOPT}{\mathrm{LOPT}}
\newcommand{\ROPT}{\mathrm{ROPT}}

\newcommand{\Xcomment}[1]{}

\newcommand{\set}[1]{\left\{#1\right\}}

\newcommand{\zo}{\set{0, 1}}

\title{Separating Hierarchical and General Hub Labelings}

\author{
Andrew V.\ Goldberg \\ Microsoft Research Silicon Valley
\and
Ilya Razenshteyn\thanks{Part of this work done while at Microsoft} \\ CSAIL, MIT
\and
Ruslan Savchenko\thanks{Part of this work done while at Microsoft} \\ Department of Mech.\ and Math., MSU
}

\date{\vspace{-5ex}}

\begin{document}

\maketitle

\begin{abstract}
In the context of distance oracles,
a labeling algorithm computes vertex labels during preprocessing.
An $s,t$ query computes the corresponding distance using the labels
of $s$ and $t$ only, without looking at the input graph.
Hub labels is a class of labels that has been extensively studied.
Performance of the hub label query depends on the label size.
Hierarchical labels are a natural special kind of hub labels.
These labels are related to other problems and can be computed
more efficiently.
This brings up a natural question of the quality of hierarchical labels.
We show that there is a gap: optimal hierarchical labels can
be polynomially bigger than the general hub labels.
To prove this result, we give tight upper and lower bounds on the size of 
hierarchical and general labels for hypercubes.
\end{abstract}

\section{Introduction}

The point-to-point shortest path problem is a fundamental optimization 
problem with many applications. Dijkstra's algorithm~\cite{d-ntpcg-59} 
solves this problem in near-linear time~\cite{g-apspa-08} on directed
and in linear time on undirected graphs~\cite{tho-99}, but some applications
require sublinear distance queries.
This is possible for some graph classes if preprocessing is allowed 
(e.g.,~\cite{dssw-erpa-09,gppr-dlg-04}). 
Peleg introduced a \emph{distance labeling} 
algorithm~\cite{pel-00} that precomputes a \emph{label} for each vertex 
such that the distance between any two vertices $s$ and $t$ can be computed 
using only their labels. 
A special case is \emph{hub labeling} (\hl)~\cite{gppr-dlg-04}: the label of $u$
consists of a collection of vertices (the {\em hubs} of $u$) with their 
distances from $u$. 
Hub labels satisfy the \emph{cover property}: 
for any two vertices $s$ and $t$, there exists a vertex $w$ on the 
shortest $s$--$t$ path that belongs to both the label of $s$ and the label of $t$.

Cohen et al.~\cite{CHKZ-03} give a polynomial-time $O(\log n)$-approximation 
algorithm for the smallest labeling (here $n$ denotes the number of vertices).
(See~\cite{BGGN-13-2} for a generalization.)
The complexity of the algorithm, however, is fairly high, making it impractical
for large graphs.
Abraham et al.~\cite{ADGW-12} introduce a class of {\em hierarchical labelings}
(\hhl) and show that \hhl\ can be computed in $O^*(nm)$ time, 
where $m$ is the number of arcs.
This makes preprocessing feasible for moderately large graphs,
and for some problem classes produces labels that are sufficiently small for
practical use.
In particular, this leads to the fastest distance oracles for
continental-size road networks~\cite{adgw-ahbla-11}.
However, the algorithm of~\cite{ADGW-12} does not have theoretical guarantees
on the size of the labels.

\hhl\ is a natural algorithm that is closely related to other widely studied
problems, such as vertex orderings for contraction 
hierarchies~\cite{gssd-chfsh-08} and elimination
sequences for chordal graphs (e.g.,~\cite{gol-80}).
This provides additional motivation for studying \hhl.
This motivation is orthogonal the relationship of \hhl\ to \hl,
which is not directly related to the above-mentioned problems.

\hhl\ is a special case of \hl, so
a natural question is how the label size is affected by
restricting the labels to be hierarchical.
In this paper we show that \hhl\ labels can be substantially bigger
than the general labels.
Note that it is enough to show this result for a special class of graphs.
We study hypercubes, which have a very simple structure.
However, proving tight bounds for them is non-trivial:
Some of our upper bound constructions and lower bound proofs are
fairly involved.

We obtain upper and lower bounds on the optimal size for both kinds of 
labels in hypercubes.
In particular, for a hypercube of dimension $d$ (with $2^d$ vertices),
we give both upper and lower bounds of $3^d$ on the \hhl\ size.
For \hl, we also give a simple construction producing labels of size $2.83^d$,
establishing a polynomial separation between the two label classes.
A more sophisticated argument based on the primal-dual method yields
$(2.5 + o(1))^d$ upper and lower bounds on the \hl\ size.
Although the upper bound proof is non-constructive, it implies that
the Cohen et al.\ approximation algorithm computes the labels of 
size $(2.5+o(1))^d$, making the bound constructive.

The paper is organized as follows.
After introducing basic definitions in Section~\ref{sec:prelim},
we prove matching upper and lower bounds on the \hhl\ size in 
Section~\ref{sec:hhl}.
Section~\ref{sec:simp} gives a simple upper bound on the size of \hl\
that is polynomially better than the lower bound on the size of \hhl.
Section~\ref{sec:pd} strengthens these results by proving a better lower
bound and a near-matching upper bound on the \hl\ size.
Section~\ref{sec:conc} contains the conclusions.

\section{Preliminaries}
\label{sec:prelim}

In this paper we consider shortest paths in an undirected graph
$G = (V,E)$, with $|V|=n$, $|E|=m$, and length $\ell(a) > 0$ for each arc $a$. 
The length of a path $P$ in $G$ is the sum of its arc lengths. 
The {\em distance query} is as follows:
given a source $s$ and a target $t$, to find the distance $\dist(s,t)$ between 
them, i.e., the length of the shortest path $P_{st}$ between $s$ and $t$ in $G$.
Often we will consider unweighted graphs ($\ell \equiv 1$).

Dijkstra's algorithm~\cite{d-ntpcg-59} solves the problem in
$O(m + n\log n)$~\cite{FT-87} time in the comparison model and in linear
time in weaker models~\cite{tho-99}.
However, for some applications, even linear time is too slow. 
For faster queries, \emph{labeling} algorithms preprocess the graph and 
store a \emph{label} with each vertex; the $s$--$t$ distance can be computed 
from the labels of $s$ and $t$. 
We study \emph{hub labelings} (\hl), a special case of the labeling method. 
For each vertex $v \in V$, \hl\ precomputes a label $L(v)$,
which contains a subset of vertices ({\em hubs}) and, for every hub $u$ 
the distance $\dist(v,u)$.
Furthermore, the labels obey the \emph{cover property}: for any two vertices 
$s$ and $t$, $L(s) \cap L(t)$ must contain at least one vertex on the 
shortest $s$--$t$ path. 

For an $s$--$t$ query, among all vertices $u \in L(s) \cap L(t)$ 
we pick the one minimizing $\dist(s,u) + \dist(u,t)$ and return 
the corresponding sum. 
If the entries in each label are sorted by hub vertex ID, 
this can be done with a sweep over the two labels, as in mergesort.
The {\em label size of $v$}, $|L(v)|$, is the number of hubs in $L(v)$. 
The time for an $s$--$t$ query is $O(|L(s)| + |L(t)|)$. 

The \emph{labeling} $L$ is the set of all labels. 
We define its \emph{size} as  $\sum_v (|L(v)|)$.
Cohen et al.~\cite{CHKZ-03} show how to generate in $O(n^4)$ time a labeling
whose size is within a factor $O(\log n)$ of the optimum.  

Given two distinct vertices $v,w$, we say that $v \preceq w$ if $L(v)$ 
contains $w$. 
A labeling is {\em hierarchical} if $\preceq$ is a partial order. 
We say that this order is {\em implied} by the labeling. 
Labelings computed by the algorithm of Cohen et al.\ are not 
necessarily hierarchical.
Given a total order on vertices, the {\em rank function} 
$r: V \rightarrow [1\ldots n]$ ranks the vertices according to the order. 
We will call the corresponding order $r$.

We define a $d$-dimensional hypercube $H=(V,E)$ graph as follows.
Let $n = 2^d$ denote the number of vertices.
Every vertex $v$ has an $d$-bit binary ID that we will also denote by $v$.
The bits are numbered from the most to the least significant one.
Two vertices $v,w$ are connected iff their IDs differ in exactly one bit.
If $i$ is the index of that bit, we say that $(v,w)$ {\em flips} $i$.
We identify vertices with their IDs, and $v \oplus w$ denotes exclusive or.
We also sometimes view vertices as subsets of $\{1 \ldots d\}$, with
bits indicating if the corresponding element is in or out of the set.
Then $v \oplus w$ is the symmetric difference.
The graph is undirected and unweighted.

\section{Tight Bounds for HHL on Hypercubes}
\label{sec:hhl}

In this section we show that a $d$-dimensional hypercube has a labeling
of size $3^d$, and this labeling is optimal.

Consider the following labeling: treat vertex IDs as sets.
$L(v)$ contains all vertices whose IDs are subsets of that of $v$.
It is easy to see that this is a valid hierarchical labeling.
The size of the labeling is
$$
\sum_{i=0}^d 2^i \binom{d}{i} = 3^d .
$$
\begin{lemma}
A $d$-dimensional hypercube has an \hhl\ of size $3^d$.
\end{lemma}

Next we show that $3^d$ is a tight bound.
Given two vertices $v$ and $w$ of the hypercube, the {\em induced hypercube}
$H_{vw}$ is the subgraph induced by the vertices that have the same bits in the
positions where the bits of $v$ and $w$ are the same, and arbitrary bits
in other positions.
$H_{vw}$ contains all shortest paths between $v$ and $w$.
For a fixed order of vertices $v_1, v_2,\ldots,v_n$ (from least to most
important), we define a {\em canonical labeling} as follows:
$w$ is in the label of $v$ iff $w$ is the maximum vertex of $H_{vw}$
with respect to the vertex order.
The labeling is valid because for any $s,t$, the maximum vertex of $H_{st}$
is in $L(s)$ and $L(t)$, and is on the $s$--$t$ shortest path.
The labeling is \hhl\ because all hubs of a vertex $v$ have ranks greater or 
equal to the rank of $v$.
The labeling is minimal because if $w$ is the maximum vertex in $H_{vw}$,
then $w$ is the only vertex of $H_{vw} \cap L(w)$, so $L(v)$ must contain $w$.

\begin{lemma}
The size of a canonical labeling is independent of the vertex ordering.
\end{lemma}
\begin{proof}
It is sufficient to show that any transposition of neighbors does not
affect the size.
Suppose we transpose $v_i$ and $v_{i+1}$.
Consider a vertex $w$.
Since only the order of $v_i$ and $v_{i+1}$ changed,
$L(w)$ can change only if either $v_i \in H_{v_{i+1}w}$ or $v_{i+1} \in H_{v_{i}w}$,
and $v_{i+1}$ is the most important vertex in the corresponding induced
hypercubes.
In the former case $v_{i+1}$ is removed from $L(w)$ after the transposition,
and in the latter case $v_i$ is added.
There are no other changes to the labels.

Consider a bijection $b: H \Rightarrow H$,
obtained by flipping all bits of $w$ in the positions
in which $v_i$ and $v_{i+1}$ differ.
We show that $v_{i+1}$ is removed from $L(w)$ iff $v_i$ is added to $L(b(w))$.
This fact implies the lemma.

Suppose $v_{i+1}$ is removed from $L(w)$, i.e., $v_i \in H_{v_{i+1}w}$ and
before the transposition $v_{i+1}$ is the maximum vertex in $H_{v_{i+1}w}$.
From $v_i \in H_{v_{i+1}w}$ it follows that $v_i$ coincides with $v_{i+1}$ in
the positions in which $v_{i+1}$ and $w$ coincide.
Thus $b$ doesn't flip bits in the positions in which $v_{i+1}$ and $w$ coincide.
So positions in which $v_{i+1}$ and $w$ coincide are exactly the same in which
$b(v_{i+1})$ and $b(w)$ coincide.
Moreover, in these positions all four $v_{i+1}$, $w$, $b(v_{i+1})$ and $b(w)$ coincide.
So each vertex from $H_{v_{i+1},w}$ contains in $H_{b(v_{i+1}),b(w)}$ and vice versa, thus implying $H_{v_{i+1}w} = H_{b(v_{i+1})b(w)}$.
Note that $b(v_{i+1}) = v_i$, and therefore $H_{v_{i+1}w} = H_{v_i b(w)}$.
Before the transposition, $v_{i+1}$ is the maximum vertex of $H_{v_i b(w)}$
and therefore $L(b(w))$ does not contain $v_i$.
After the transposition, $v_i$ becomes the maximum vertex, so
$L(b(w))$ contains $v_i$.

This proves the if part of the claim. The proof of the only if part is similar.
\end{proof}

The hierarchical labeling of size $3^d$ defined above
is canonical if the vertices are ordered in the reverse order of their IDs.
Therefore we have the following theorem.
\begin{theorem}
Any hierarchical labeling of a hypercube has size of at least $3^d$.
\end{theorem}

\section{An $O(2.83^d)$ HL for Hypercubes}
\label{sec:simp}

Next we show an \hl\ for the hypercube of size $O(2.83^d)$.
Combined with the results of Section~\ref{sec:hhl}, this implies
that there is a polynomial (in $n=2^d$) gap between hierarchical
and non-hierarchical label sizes.

Consider the following \hl\ $L$:
For every $v$, $L(v)$ contains all vertices with the first $\lfloor d/2\rfloor$ bits
of ID identical to those of $v$ and the rest arbitrary, and all
vertices with the last $\lceil d/2\rceil$ bits of ID identical to those of $v$ and the
rest arbitrary.
It is easy to see that this labeling is non-hierarchical.
For example, consider two distinct vertices $v,w$ with the 
same $\lfloor d/2\rfloor$ first ID bits.
Then $v \in L(w)$ and $w \in L(v)$.

To see that the labeling is valid, fix $s,t$ and consider a vertex $u$
with the first $\lfloor d/2\rfloor$ bits equal to $t$ and the last $\lceil d/2\rceil$ bits equal to $s$.
Clearly $u$ is in $L(s) \cap L(t)$.
The shortest path that first changes bits of the first half of $s$ to 
those of $t$ and then the last bits passes through $u$.

The size of the labeling is 
$2^d \cdot (2^{\lfloor d/2\rfloor} + 2^{\lceil d/2\rceil}) = O(2^{\frac{3}{2}d}) = O(2.83^d)$.
We have the following result.
\begin{theorem}\label{tm:simp}
A $d$-dimensional hypercube has an \hl\ of size $O(2.83^d)$.
\end{theorem}

\section{Better HL Bounds}
\label{sec:pd}

The bound of Theorem~\ref{tm:simp} can be improved.
Let $\OPT$ be the optimal hub labeling size for a $d$-dimensional hypercube.
In this section we prove the following result.

\begin{theorem}
  $\OPT = (2.5 + o(1))^d$
\end{theorem}

The proof uses the primal-dual method.
Following~\cite{CHKZ-03}, we view the labeling problem as a special case
of \textbf{SET-COVER}.
We state the problem of finding an optimal hub labeling of a hypercube 
as an integer linear program (ILP) which is a special case of a standard 
ILP formulation of \textbf{SET-COVER} (see e.g.~\cite{vaz-01}), 
with the sets corresponding
to the shortest paths in the hypercube.
For every vertex $v \in \zo^d$ and every subset $S \subseteq \zo^d$ we introduce a binary variable
$x_{v,S}$. In the optimal solution $x_{v,S} = 1$ iff $S$ is the set of vertices whose labels contain $v$.
For every \emph{unordered} pair of vertices $\set{i, j} \subseteq \zo^d$ we introduce the following constraint:
there must be a vertex $v \in \zo^d$ and a subset $S \subseteq \zo^d$ such that $v \in H_{i j}$
(recall that the subcube $H_{ij}$ consists of vertices that lie on the shortest paths from $i$ to $j$),
$\set{i, j} \subseteq S$, and $x_{v,S} = 1$.
Thus, $\OPT$ is the optimal value of the following integer linear program:
\begin{equation*}
\;\;\;\; \min\sum_{v,S} |S| \cdot x_{v,S} \;\;\;\;\;\textrm{ subject to}
\end{equation*}
\vspace{-20pt}
\begin{equation}
\label{ilp}
\begin{cases}
   x_{v,S} \in \zo & \forall \; v \in \zo^d , S \subseteq \zo^d \\
  \sum_{\begin{smallmatrix}S \supseteq \set{i, j}\\ v \in H_{ij} \end{smallmatrix}}
  x_{v,S} \geq 1 & \forall \set{i, j} \subseteq \zo^d 
\end{cases}
\end{equation}
We consider the following LP-relaxation of~(\ref{ilp}):
\begin{equation*}
\;\;\;\; \min\sum_{v,S} |S| \cdot x_{v,S} \;\;\;\;\;\textrm{ subject to}
\end{equation*}
\vspace{-20pt}
\begin{equation}
\label{lp}
\begin{cases}
  x_{v,S} \geq 0 & \forall \; v \in \zo^d , S \subseteq \zo^d \\
  \sum_{\begin{smallmatrix}S \supseteq \set{i, j}\\ v \in H_{ij} \end{smallmatrix}}
  x_{v,S} \geq 1 & \forall \set{i, j} \subseteq \zo^d
\end{cases}
\end{equation}
We denote the optimal value of~(\ref{lp}) by $\LOPT$,
and bound $\OPT$ as follows:
\begin{lemma}
  \label{rounding}
  $\LOPT \leq \OPT \leq O(d) \cdot \LOPT$
\end{lemma}
\begin{proof}
  The first inequality follows from the fact that~(\ref{lp}) is a relaxation of~(\ref{ilp}). 
  
  As~(\ref{ilp}) corresponds to the standard ILP-formulation of \textbf{SET-COVER},
  and~(\ref{lp}) is the standard LP-relaxation for it, we can use
  the well-known (e.g., \cite{vaz-01}, Theorem 13.3) result:
  The integrality gap of LP-relaxation
  for \textbf{SET-COVER} is logarithmic in the number of elements we want to cover, which in our case is $O(n^2) = O(2^{2d})$.
  This implies the second inequality.
\end{proof}

Now consider the dual program to~(\ref{lp}).
\begin{equation*}
\;\;\;\; \max\sum_{\{i, j\}} y_{\{i, j\}} \;\;\;\;\;\textrm{ subject to}
\end{equation*}
\vspace{-20pt}
\begin{equation}
\label{dual}
\begin{cases}
 y_{\{i, j\}} \geq 0 & \forall \set{i, j} \subseteq \zo^d \\
 \sum_{\begin{smallmatrix}\set{i, j} \subseteq S \\ H_{ij} \ni v \end{smallmatrix}}
  y_{\{i, j\}} \leq |S| & \forall v \in \zo^d, S \subseteq \zo^d
\end{cases}
\end{equation}
The dual problem is a path packing problem.
The strong duality theorem implies that $\LOPT$ is  also the optimal solution 
value for~(\ref{dual}).

    We strengthen~(\ref{dual}) by requiring that the values $y_{\{i,j\}}$ depend only on the distance between $i$ and $j$.
    Thus, we have variables $\tilde{y}_0, \tilde{y}_1, \ldots, \tilde{y}_d$.
    Let $N_k$ denote the number of vertex pairs at distance $k$ from each other.
    Note that since $\tilde{y}$'s depend only on the distance and the hypercube is symmetric, it is enough to add
    constraints only for one vertex (e.g., $0^d$); other constraints are redundant. We have the following linear program,
    which we call \emph{regular}, and denote its optimal value by $\ROPT$.
\begin{equation*}
\;\;\;\; \max\sum_{k} N_k \cdot \tilde{y}_{k} \;\;\;\;\;\textrm{ subject to}
\end{equation*}
\vspace{-20pt}
\begin{equation}
  \label{regular}
  \begin{cases}
    \tilde{y}_{k} \geq 0 & \forall \; 0 \leq k \leq d \\
    \sum_{\begin{smallmatrix}\set{i, j} \subseteq S \\ H_{i j} \ni 0^d \end{smallmatrix}}
    \tilde{y}_{\mathrm{dist}(i, j)} \leq |S| & \forall S \subseteq \zo^d
  \end{cases}
\end{equation}
Clearly $\ROPT \leq \LOPT$. The following lemma shows 
that in fact the two values are the same.

\begin{lemma} $\ROPT \geq \LOPT$
  \label{projection}
\end{lemma}
\begin{proof}
  Intuitively, the proof shows that by averaging a solution for~(\ref{dual}), we obtain a feasible solution
  for~(\ref{regular}) with the same objective function value.
  
  Given a feasible solution $y_{\{i,j\}}$ for~(\ref{dual}), define
  $$
  \tilde{y}_k = \frac{\sum_{\set{i, j} : \mathrm{dist}(i, j) = k} y_{\{i,j\}}}{N_k}.
  $$
  From the definition,
  $$
  \sum_{\{i, j\}} y_{\{i, j\}} = \sum_k N_k \cdot \tilde{y}_k.
  $$
  We need to show that $\tilde{y}_k$ is a feasible solution for~(\ref{regular}).
  
  Consider a random mapping $\varphi \colon \zo^d \to \zo^d$ that is a composition of a mapping $i \mapsto i \oplus p$,
  where $p \in \zo^d$ is a uniformly random vertex, and a uniformly random permutation of coordinates.
  Then, clearly, we have the following properties:
  \begin{itemize}
  \item $\varphi$ preserves distance;
  \item  $\varphi$ is a bijection;
  \item if the distance between $i$ and $j$ is $k$, then the pair $(\varphi(i), \varphi(j))$ is
    uniformly distributed among all pairs of vertices at distance $k$ from each other.
  \end{itemize}
  
  Let $S \subseteq \zo^d$ be a fixed subset of vertices.
  As $y_{\{i,j\}}$ is a feasible solution of~(\ref{dual}), we have
  $$
  \sum_{\begin{smallmatrix}\set{i, j} \subseteq S\\H_{i j} \ni 0^d\end{smallmatrix}} y_{\{i,j\}} \leq |S|.
  $$
  We define a random variable $X$ as follows:
  $$
  X = \sum_{\begin{smallmatrix}\set{i, j} \subseteq \varphi(S)\\H_{i j} \ni \varphi(0^d)\end{smallmatrix}} y_{\{i, j\}}.
  $$
  Since $\varphi$ is a bijection and $y$ is a feasible solution of~(\ref{dual}), we have
  $\mathbf{E}_{\varphi}[X] \leq |S|$.
  Furthermore, $\mathbf{E}_{\varphi}[X]$ is equal to
  $$
  \mathbf{E}_{\varphi}[
    \sum_{\begin{smallmatrix}\set{i, j} \subseteq \varphi(S)\\H_{i j} \ni \varphi(0^d)\end{smallmatrix}}
    y_{\{i, j\}}] = \mathbf{E}_{\varphi}[
    \sum_{\begin{smallmatrix}\set{i, j} \subseteq S\\H_{ij} \ni 0^d\end{smallmatrix}}
    y_{\{\varphi(i), \varphi(j)\}}
    ] = 
  \sum_{\begin{smallmatrix}\set{i, j} \subseteq S\\H_{ij} \ni 0^d\end{smallmatrix}}
  \mathbf{E}_{\varphi}\left[y_{\{\varphi(i), \varphi(j)\}}\right].
  $$
  Since $(\varphi(i), \varphi(j))$ is uniformly distributed among all pairs of vertices at distance
  $\mathrm{dist}(i, j)$,
  the last expression is equal to
  $\sum_{\begin{smallmatrix}\set{i, j} \subseteq S \\ H_{i j} \ni 0^d \end{smallmatrix}}
  \tilde{y}_{\mathrm{dist}(i, j)}$.
\end{proof}
Combining Lemmas~\ref{rounding} and~\ref{projection}, we get
$$
\ROPT \leq \OPT \leq O(d) \cdot \ROPT.
$$
It remains to prove that $\ROPT = (2.5 + o(1))^d$.
For $0 \leq k \leq d$, let $\tilde{y}_k^*$ denote the maximum feasible value of $\tilde{y}_k$. It is easy to see that
$\max_k N_k\tilde{y}_k^* \leq \ROPT \leq (d + 1) \cdot \max_k N_k\tilde{y}_k^*$.
Next we show that $\max_k N_k\tilde{y}_k^* = (2.5 + o(1))^d$.

To better understand (\ref{regular}), consider the
graphs $G_k$ for $0 \le k \le d$.
Vertices of $G_k$ are the same as those of the hypercube,
interpreted as subsets of $\{1, \ldots, d\}$.
Two vertices are connected by an edge in $G_k$
iff there is a shortest path of length $k$
between them that passes through $0^d$ in the hypercube.
This holds iff the corresponding subsets are disjoint
and the cardinality of the union of the subsets is equal to $k$.

Consider connected components of $G_k$.
By $C_k^i$ ($0 \leq i \leq \lfloor k / 2\rfloor$) we denote the component
that contains sets of cardinality $i$ (and $k - i$).

If $k$ is odd or $i \not = k/2$, $C^i_k$ is a bipartite graph,
with the right side vertices corresponding to sets of cardinality $i$,
and the left side vertices -- to sets of cardinality $k-i$.
The number of these vertices is $\binom{d}{i}$ and $\binom{d}{k-i}$, 
respectively.
$C^i_k$ is a regular bipartite graph with vertex degree on the right side
equal to $\binom{d-i}{k-i}$: given a subset of $i$ vertices, this is the
number of ways to choose a disjoint subset of size $k-i$.
The density of $C^i_k$ is
$$
\frac{\binom{d}{i}\cdot \binom{d-i}{k-i}}{\binom{d}{i} + \binom{d}{k-i}} .
$$

If $k$ is even and $i = k/2$, then $G^i_k$ is a graph with
$\binom{d}{i}$ vertices corresponding to the subsets of size $i$.
The graph is regular, with the degree $\binom{d-i}{k-i}$.
The density of $C^i_k$ written to be consistent with the previous case is again
$$
\frac{\binom{d}{i}\cdot\binom{d-i}{k-i}}{\binom{d}{i} + \binom{d}{k-i}} .
$$

Next we prove a lemma about regular graphs, which may be of 
independent interest.

\begin{lemma}\label{lm:reg}
In a regular graph, density of any subgraph does not exceed the density of
the graph.
In a regular bipartite graph (i.e., degrees of each part are uniform),
the density of any subgraph does not exceed the density of
the graph.
\end{lemma}
\begin{proof}
Let $x$ be the degree of a regular graph.
The density is a half of the average degree, and the average degree
of any subgraph is at most $x$, so the lemma follows.

Now consider a bipartite graph with $X$ vertices on the left side and 
$Y$ vertices of the right side.
Consider a subgraph with $X'$ vertices on the left and $Y'$ vertices on the
right.
Assume $X/X' \ge Y/Y'$; the other case is symmetric.

Let $x$ be the degree of the vertices on the left size, then the graph density
is $x\cdot X/(X + Y)$.
For the subgraph, the number of edges adjacent to $X'$ is at most $x\cdot X'$,
so the subgraph density is at most
$$
\frac{x\cdot X'}{X' + Y'} =
\frac{x\cdot X}{X + Y'X/X'} \le
\frac{x\cdot X}{X + Y'Y/Y'} =
\frac{x\cdot X}{X + Y} .
$$
\end{proof}

By the lemma, each $C^i_k$ is the densest subgraph of itself,
and since $C^i_k$ are connected components of $G_k$, the
densest $C^i_k$ is the densest subgraph of $G_k$.

Next we prove a lemma that gives (the inverse of) the value of maximum
density of a subgraph of $G_k$.
\begin{lemma}
  \label{lm:middle}
  For fixed $d$ and $k$ with $k \le d$, the minimum of the expression
  $$
  \frac{\binom{d}{x} + \binom{d}{k-x}}
  {\binom{d}{x}\cdot\binom{d-x}{k-x}}
  $$
  is achieved for $x = \lfloor k/2 \rfloor$ and $x = \lceil k/2 \rceil$
  (with the two values being equal).
\end{lemma}
\begin{proof}
Using the standard identity
$$
\binom{d}{x}\cdot\binom{d-x}{k-x} = \binom{d}{k - x}\cdot\binom{d-k+x}{x}
$$
we write the expression in the lemma as
$$
\frac{1}{\binom{d-x}{k-x}} + \frac{1}{\binom{d-k+x}{x}} = 
\frac{(d-k)!(k-x)!}{(d-x)!} + \frac{(d-k)! x!}{(d-k+x)!} .
$$
Since $d-k$ is a constant, we need to minimize
\begin{equation}\label{eq:to_min}
\frac{1}{(k-x+1)\cdot \ldots \cdot (d-x)} + 
\frac{1}{(x+1)\cdot \ldots \cdot (d-k+x)} .
\end{equation}
Note that the expression is symmetric around 
$x = k/2$: for $y = k - x$, the expression
becomes
$$
\frac{1}{(y+1)\cdot \ldots \cdot (d-k+y)} +
\frac{1}{(k-y+1)\cdot \ldots \cdot (d-y)} .
$$
So it is enough to show that for $x \ge \lceil k/2 \rceil$, the
minimum is achieved at $x = \lceil k/2 \rceil$.

We will need the following auxiliary lemma.

\begin{lemma}
    \label{manip}
    If $0 \leq s \leq t$ and $\alpha \geq \beta \geq 1$, then $\alpha t + s / \beta \geq t + s$.
\end{lemma}
\begin{proof}
    Since $2 \leq \alpha + 1 / \alpha \leq \alpha + 1 / \beta$, we have $\alpha - 1 \geq 1 - 1 / \beta$.
    Thus, $(\alpha - 1) t \geq s(1 - 1 / \beta)$, and the lemma follows.
\end{proof}

It is clear that for every $x$ the first term of (\ref{eq:to_min}) is not less than the second one.
If we move from $x$ to $x + 1$, then the first term is multiplied by $(d - x)/(k - x)$, and the second
term is divided by $(d - k + x + 1) / (x + 1)$. Since
$$
    \frac{d - x}{k - x} - \frac{d - k + x + 1}{x + 1} = \frac{(d - k) (2x + 1 - k)}{(x + 1)(k - x)} \geq 0,
$$
we can invoke Lemma~\ref{manip} with $t$ and $s$ being equal to the first and the second term of (\ref{eq:to_min}),
respectively, $\alpha = (d - x)/(k - x)$, $\beta = (d - k + x + 1) / (x + 1)$.
\end{proof}

Recall that $\tilde{y}_k^*$ denotes the maximum feasible value of $\tilde{y}_k$.

\begin{lemma}\label{lm:binom}
  $$
  \tilde{y}_k^* = \begin{cases}
    1 & k = 0 \\
    2 / \binom{d - i}{i} & k = 2i, i > 0 \\
    \left(\binom{d}{i} + \binom{d}{i + 1}\right) / \left(\binom{d}{i} \cdot \binom{d - i}{i + 1}\right) & k = 2i + 1.
  \end{cases}
  $$
\end{lemma}
\begin{proof}
Fix $k$ and consider the maximum density subgraph of $G_k$.
Inverse of the subgraph density is an upper bound on a feasible value
of $\tilde{y}_k$.

On the other hand, it is clear that we can set $\tilde{y}_k$ to the inverse
density of the densest subgraph of $G_k$ and other $\tilde{y}$'s to zero, and
obtain the feasible solution of (\ref{regular}).

By applying Lemma~\ref{lm:middle}, we obtain the desired statement.
\end{proof}

Recall that $N_k$ denotes the number of vertex pairs at distance $k$ 
from each other.
For each vertex $v$, we can choose a subset of $k$ bit positions and flip bits
in these positions, obtaining a vertex at distance $k$ from $v$.
This counts the ordered pairs, we need to divide by two to get the the number of unordered pairs:
$$
N_k = 2^d \binom{d}{k} /2,
$$
except for the case $k = 0$, where $N_0 = 2^d$.

Finally, we need to find the maximum value of
$$
\psi(k) := N_k \cdot \tilde{y}_k^* = 2^{d} \cdot \begin{cases}
  \binom{d}{2i} / \binom{d - i}{i} & k = 2i \\
  \binom{d}{2i + 1} \cdot \left(\binom{d}{i} + \binom{d}{i + 1}\right) / 
  \left(2 \cdot \binom{d}{i} \cdot \binom{d - i}{i + 1}\right) & k = 2i + 1.
\end{cases}
$$

One can easily see that $\psi(2i + 1) / \psi(2i) = (d+1)/(4i+2)$.
So, if we restrict our attention to the case $k = 2i$, we could potentially lose only polynomial
factors.

We have
$$
    \frac{\psi(2i + 2)}{\psi(2i)} = \frac{d-i}{4i+2}.
$$
This expression is greater than one if $i < (d-2)/5$. 
The optimal $i$ has to be as close as possible to the bound. 
As $d \to \infty$, this is $\frac{d}{5}\cdot(1+o(1))$.

We will use the standard fact: if for $n \to \infty, m / n \to \alpha$, then
$$
    \binom{n}{m} = (2^{H(\alpha)} + o(1))^n,
$$
where $H$ is the Shannon entropy function 
$H(\alpha) = -\alpha \log_2 \alpha - (1 - \alpha) \log_2 (1-\alpha)$.

Thus, if $d \to \infty, k / d \to 2/5$, then
$$
    \psi(k) = (2^{1 + H(0.4) - 0.8 \cdot H(0.25)} + o(1))^d.
$$
One can verify that
$$
    2^{1 + H(0.4) - 0.8 \cdot H(0.25)} = 2.5,
$$
so we have the desired result.

\section{Concluding Remarks}
\label{sec:conc}

We show a polynomial gap between the sizes of \hl\ and \hhl\ 
for hypercubes.
Although our existence proof for $(2.5 + o(1))^d$-size \hl\ is non-constructive,
the approximation algorithm of~\cite{CHKZ-03} can build such labels
in polynomial time.
However, it is unclear how these labels look like.
It would be interesting to have an explicit construction of such labels.

Little is known about the problem of computing the smallest \hhl.
We do not know if the problem is NP-hard, and we know no 
polynomial-time algorithm for it (exact or polylog-approximate).
These are interesting open problems.

The \hl\ vs.\ \hhl\ separation we show does not mean that \hhl\ labels
are substantially bigger than the \hl\ ones for any graphs.
In particular, experiments suggest that \hhl\ works well for road networks.
It would be interesting to characterize the class of networks for which
\hhl\ works well.

Note that an arbitrary (non-hub) labelings for the hypercube can be small:
we can compute the distances from the standard $d$-bit vertex IDs.
It would be interesting to show the gap between \hl\ and \hhl\ for graph classes
for which arbitrary labelings must be big.

We believe that one can prove an $\Theta^*(n^{1.5})$ bound for \hl\ size on
constant degree random graphs using the primal-dual method.
However, for this graphs it is unclear how to prove tight bounds on the size of
\hhl.


\begin{thebibliography}{10}

\bibitem{ADGW-12}
I.~Abraham, D.~Delling, A.V. Goldberg, and R.F. Werneck.
\newblock {Hierarchical Hub Labelings for Shortest Paths}.
\newblock In {\em Proc.\ 20th European Symposium on Algorithms (ESA 2012)},
  2012.

\bibitem{adgw-ahbla-11}
Ittai Abraham, Daniel Delling, Andrew~V. Goldberg, and Renato~F. Werneck.
\newblock {A Hub-Based Labeling Algorithm for Shortest Paths on Road Networks}.
\newblock In Panos~M. Pardalos and Steffen Rebennack, editors, {\em Proceedings
  of the 10th International Symposium on Experimental Algorithms (SEA'11)},
  volume 6630 of {\em Lecture Notes in Computer Science}, pages 230--241.
  Springer, 2011.

\bibitem{BGGN-13-2}
M.~Babenko, A.~V. Goldberg, A.~Gupta, and V.~Nagarajan.
\newblock {Algorithms for Hub Label Optimization}.
\newblock In {\em Proc. 30th ICALP}. Springer-Verlag, 2013.

\bibitem{CHKZ-03}
E.~Cohen, E.~Halperin, H.~Kaplan, and U.~Zwick.
\newblock {Reachability and Distance Queries via 2-hop Labels}.
\newblock {\em SIAM Journal on Computing}, 32, 2003.

\bibitem{dssw-erpa-09}
Daniel Delling, Peter Sanders, Dominik Schultes, and Dorothea Wagner.
\newblock {Engineering Route Planning Algorithms}.
\newblock In J{\"u}rgen Lerner, Dorothea Wagner, and Katharina Zweig, editors,
  {\em Algorithmics of Large and Complex Networks}, volume 5515 of {\em Lecture
  Notes in Computer Science}, pages 117--139. Springer, 2009.

\bibitem{d-ntpcg-59}
Edsger~W. Dijkstra.
\newblock {A Note on Two Problems in Connexion with Graphs}.
\newblock {\em Numerische Mathematik}, 1:269--271, 1959.

\bibitem{FT-87}
M.~L. Fredman and R.~E. Tarjan.
\newblock {Fibonacci Heaps and Their Uses in Improved Network Optimization
  Algorithms}.
\newblock {\em J. Assoc. Comput. Mach.}, 34:596--615, 1987.

\bibitem{gppr-dlg-04}
Cyril Gavoille, David Peleg, St\'{e}phane P\'{e}rennes, and Ran Raz.
\newblock {Distance Labeling in Graphs}.
\newblock {\em Journal of Algorithms}, 53:85--112, 2004.

\bibitem{gssd-chfsh-08}
Robert Geisberger, Peter Sanders, Dominik Schultes, and Daniel Delling.
\newblock {Contraction Hierarchies: Faster and Simpler Hierarchical Routing in
  Road Networks}.
\newblock In Catherine~C. McGeoch, editor, {\em Proceedings of the 7th
  International Workshop on Experimental Algorithms (WEA'08)}, volume 5038 of
  {\em Lecture Notes in Computer Science}, pages 319--333. Springer, June 2008.

\bibitem{g-apspa-08}
Andrew~V. Goldberg.
\newblock {A Practical Shortest Path Algorithm with Linear Expected Time}.
\newblock {\em SIAM Journal on Computing}, 37:1637--1655, 2008.

\bibitem{gol-80}
M.C. Golumbic.
\newblock {\em {Algorithmic Graph Theory and Perfect Graphs}}.
\newblock Academic Press, New York, 1980.

\bibitem{pel-00}
David Peleg.
\newblock Proximity-preserving labeling schemes.
\newblock {\em Journal of Graph Theory}, 33(3):167--176, 2000.

\bibitem{tho-99}
M.~Thorup.
\newblock {Undirected Single-Source Shortest Paths with Positive Integer
  Weights in Linear Time}.
\newblock {\em J. Assoc. Comput. Mach.}, 46:362--394, 1999.

\bibitem{vaz-01}
Vijay~V. Vazirani.
\newblock {\em {Approximation Algorithms}}.
\newblock Springer, 2001.

\end{thebibliography}

\end{document}